\documentclass[11pt]{article}

\usepackage[a4paper,margin=30mm]{geometry}
\usepackage[T1]{fontenc}
\usepackage[utf8]{inputenc}
\usepackage{lmodern}
\usepackage{microtype}
\usepackage{amsmath,amssymb,amsthm,mathtools}
\usepackage{enumitem}
\usepackage{xcolor}
\usepackage{hyperref}
\usepackage[nameinlink,capitalize]{cleveref}

\hypersetup{
  colorlinks=true,
  linkcolor=blue!55!black,
  citecolor=blue!55!black,
  urlcolor=blue!55!black,
  pdftitle={Flux-explicit Cheeger bounds for magnetic Laplacians on compact metric graphs},
  pdfauthor={B. E. Dahlke}
}

\setlist[itemize]{leftmargin=1.4em}

\theoremstyle{plain}
\newtheorem{theorem}{Theorem}[section]
\newtheorem{proposition}[theorem]{Proposition}

\newtheorem{corollary}[theorem]{Corollary}

\theoremstyle{definition}
\newtheorem{definition}[theorem]{Definition}
\newtheorem{example}[theorem]{Example}

\theoremstyle{remark}
\newtheorem{remark}[theorem]{Remark}

\newcommand{\dd}{\mathrm d}
\newcommand{\ii}{\mathrm i}
\newcommand{\ee}{\mathrm e}
\newcommand{\R}{\mathbb R}
\newcommand{\Z}{\mathbb Z}

\newcommand{\cR}{\mathcal R}
\newcommand{\cC}{\mathcal C}

\newcommand{\dist}{\operatorname{dist}}

\newcommand{\dom}{\operatorname{dom}}

\newcommand{\norm}[1]{\lVert #1\rVert}

\newcommand{\qA}{\mathfrak q_A}

\crefname{theorem}{Theorem}{Theorems}
\Crefname{theorem}{Theorem}{Theorems}
\crefname{proposition}{Proposition}{Propositions}
\Crefname{proposition}{Proposition}{Propositions}
\crefname{lemma}{Lemma}{Lemmas}
\Crefname{lemma}{Lemma}{Lemmas}
\crefname{corollary}{Corollary}{Corollaries}
\Crefname{corollary}{Corollary}{Corollaries}
\crefname{definition}{Definition}{Definitions}
\Crefname{definition}{Definition}{Definitions}
\crefname{example}{Example}{Examples}
\Crefname{example}{Example}{Examples}

\begin{document}

\begin{center}
{\fontsize{17}{21}\selectfont\bfseries Flux-explicit Cheeger bounds for magnetic Laplacians\\[2mm]
on compact metric graphs}\par
\vspace{6mm}
{\normalsize B. E. Dahlke\footnote{Email: \href{mailto:kontakt@bjoern-dahlke.de}{kontakt@bjoern-dahlke.de}. ORCID: 0009-0003-3861-1800.}}\par
\vspace{2mm}
{\normalsize May 25, 2026}
\end{center}

\vspace{4mm}

\begin{abstract}
Let $\Gamma$ be a finite compact connected metric graph and let
$A\in L^\infty(\Gamma)$ be a real magnetic potential.  The magnetic
Laplacian $H_A$ with standard vertex conditions is defined by the closed
quadratic form
\[
        \mathfrak q_A[u]=\sum_{e}\int_e |(-\ii\partial_x-A_e)u_e|^2\,\dd x .
\]
A magnetic Cheeger constant is introduced by adding to the usual boundary
term the frustration index of the potential on subgraphs.  The first point of
the paper is that, on a metric graph, this frustration index is exactly a
finite-dimensional $\ell^1$ flux distance determined by the periods of $A$ on
cycles.  Consequently the Cheeger constant can be written directly in terms of
Aharonov--Bohm fluxes.  We prove
\[
        \lambda_0(H_A)\geq \frac{1}{8}\,h_A(\Gamma)^2,
\]
where $\lambda_0(H_A)$ is the bottom of the spectrum.  In particular, if
$\delta_\Gamma(\Phi)$ denotes the flux distance of the global cycle-flux vector
from the integral flux lattice, then
\[
        \lambda_0(H_A)\geq \frac{1}{8}
        \min\left\{h_{\rm pr}(\Gamma),\frac{\delta_\Gamma(\Phi)}{|\Gamma|}\right\}^2 .
\]
The estimate gives explicit $L^2$ decay bounds for the magnetic heat semigroup
and for its magnetic energy.  The constants are not asserted to be sharp; the
emphasis is on the flux dependence and on the self-contained metric-graph
formulation.
\end{abstract}

\vspace{3mm}

\clearpage

\section{Introduction}

Metric graphs provide a one-dimensional model for wave and diffusion phenomena
on thin networks.  In the quantum-graph setting the basic operator is the
second derivative on each edge, coupled by vertex conditions.  Standard
references include the monograph of Berkolaiko and Kuchment \cite{BerkolaikoKuchment2013}
and the survey of Kuchment \cite{Kuchment2008}.  The corresponding spectral
geometry has developed in several directions, among them isoperimetric bounds
for eigenvalues.  Cheeger-type inequalities for the standard Laplacian on
metric graphs go back to Nicaise \cite{Nicaise1987}; they are discussed, for
example, by Post \cite{Post2009} and by Kennedy and Mugnolo \cite{KennedyMugnolo2016}.
Related estimates for spectral gaps and connectivity appear in
\cite{Kurasov2013,BerkolaikoKennedyKurasovMugnolo2017}.

The presence of a magnetic potential changes the bottom of the spectrum only
through its gauge class.  On a compact graph this gauge class is encoded by the
periods of the potential around a basis of cycles.  The resulting flux
parameters are the graph-theoretic analogue of Aharonov--Bohm fluxes.  The
spectral dependence of quantum graphs on magnetic perturbations is a standard
part of the theory; see, for instance, Berkolaiko and Weyand
\cite{BerkolaikoWeyand2014}.  On the discrete side, magnetic Laplacians were
studied by Shubin \cite{Shubin1994}, by Sunada \cite{Sunada1994}, and by Colin
de Verdiere, Torki-Hamza and Truc \cite{ColinTorkiTruc2011}, among others.
For compact Riemannian manifolds, magnetic spectral estimates of this type go
back in particular to Shigekawa \cite{Shigekawa1987}.  Lange, Liu,
Peyerimhoff and Post \cite{LangeLiuPeyerimhoffPost2015} introduced a Cheeger
constant involving a frustration index and proved magnetic Cheeger inequalities
for graph Laplacians with signatures, discrete magnetic Laplacians and magnetic
Laplacians on closed Riemannian manifolds.

The purpose of the present note is more specific.  For compact metric graphs we
write the magnetic frustration index directly in flux coordinates.  The result
is a Cheeger lower bound whose quantities are computable from cuts, lengths and
cycle fluxes.  The subsequent semigroup estimate is then immediate from the
spectral theorem, but it gives a useful formulation: non-integral fluxes open a
spectral gap and hence force exponential $L^2$ decay of the magnetic heat flow.
The gap closes precisely when the flux vector belongs to $2\pi\Z^\beta$, where
$\beta$ is the first Betti number.  For a single circle this is exactly the
usual Aharonov--Bohm periodicity; the largest gap occurs at half-periodic flux.

The note is deliberately local in scope.  It does not claim a new abstract
magnetic Cheeger mechanism beyond the frustration-index framework of
\cite{LangeLiuPeyerimhoffPost2015}.  The contribution is the metric-graph
identification of the frustration term with an explicit finite-dimensional flux
norm and the consequent heat-flow bounds.  The numerical constant in the
Cheeger inequality is not optimized.

\section{Magnetic Laplacians and fluxes}

Let $\Gamma=(V,E,\ell)$ be a finite compact connected metric graph.  Each edge
$e\in E$ is identified with an interval $(0,\ell_e)$, where $\ell_e>0$, and
\[
        L_\Gamma:=|\Gamma|:=\sum_{e\in E}\ell_e
\]
denotes the total length.  Multiple edges and loops are allowed.  An arbitrary
orientation is fixed on every edge.  The Hilbert space is
\[
        L^2(\Gamma):=\bigoplus_{e\in E}L^2(0,\ell_e).
\]
The standard first-order Sobolev space is
\[
 H^1(\Gamma):=\left\{u=(u_e)_{e\in E}: u_e\in H^1(0,\ell_e),
          \ u \text{ is continuous at every vertex}\right\} .
\]

Let $A=(A_e)_{e\in E}$ be real-valued with $A_e\in L^\infty(0,\ell_e)$.  We use
the sign convention
\[
        D_A u :=(-\ii\partial_x-A)u .
\]
The magnetic quadratic form is
\begin{equation}\label{eq:qA}
        \qA[u]:=\sum_{e\in E}\int_0^{\ell_e}|D_Au_e(x)|^2\,\dd x,
        \qquad \dom\qA=H^1(\Gamma).
\end{equation}
Since $A\in L^\infty$, this form is closed and lower semibounded.  Its
Friedrichs operator is denoted by $H_A$.  The operator acts as
$(-\ii\partial_x-A_e)^2$ on each edge, with continuity at vertices and magnetic
Kirchhoff conditions
\begin{equation}\label{eq:mag-kirchhoff}
        \sum_{e\sim v}\nu_{v,e}\bigl(D_Au_e\bigr)(v)=0,
        \qquad v\in V,
\end{equation}
where $\nu_{v,e}=1$ if the chosen coordinate on $e$ points away from $v$ and
$\nu_{v,e}=-1$ if it points into $v$.  The spectrum is discrete.  Its lowest
eigenvalue is
\begin{equation}\label{eq:lambda0}
        \lambda_0(A):=\lambda_0(H_A)
        =\inf_{0\neq u\in H^1(\Gamma)}\frac{\qA[u]}{\norm{u}_{L^2(\Gamma)}^2} .
\end{equation}

The first Betti number is
\[
        \beta=\beta_1(\Gamma)=|E|-|V|+1 .
\]
Choose an oriented integral cycle basis
\[
        \cC=(\gamma_1,\ldots,\gamma_\beta).
\]
This basis is chosen after inserting artificial vertices if necessary.  The cycle matrix associated
with this basis is the matrix $C=(\sigma_{je})\in\{-1,0,1\}^{\beta\times E}$,
where $\sigma_{je}=1$ if $\gamma_j$ traverses $e$ with its orientation,
$\sigma_{je}=-1$ if it traverses it against the orientation, and
$\sigma_{je}=0$ otherwise.  The flux vector of $A$ relative to $\cC$ is
\begin{equation}\label{eq:fluxvector}
        \Phi(A)=\bigl(\Phi_1(A),\ldots,\Phi_\beta(A)\bigr)
        \in \R^\beta/(2\pi\Z)^\beta,
\end{equation}
where
\begin{equation}\label{eq:fluxcomponent}
        \Phi_j(A)=\sum_{e\in E}\sigma_{je}\int_0^{\ell_e}A_e(x)\,\dd x
        \quad \text{modulo }2\pi .
\end{equation}
If $\beta=0$, the vector is absent and every magnetic potential is gauge
removable.  In general, two potentials with the same flux vector are unitarily
equivalent by multiplication with a single-valued phase.  Thus the spectrum of
$H_A$ depends on $A$ only through the class \eqref{eq:fluxvector}.

\section{Frustration index in flux coordinates}

The magnetic Cheeger quantity used below is defined on regular domains.  A
regular domain $U\subset\Gamma$ is a non-empty open subset obtained, after
cutting finitely many points, as a finite union of open intervals and full
edge-pieces.  The case $U=\Gamma$ is included.  Its length is denoted by $|U|$.
The boundary size $|\partial U|$ is the number of loose endpoints created by
cutting, counted with multiplicity in the cut graph.  This convention agrees
with the one-dimensional coarea formula.  For example, removing one point from
a circle creates two boundary endpoints.

The completion of $U$ after cutting is again a finite compact metric graph,
possibly disconnected and with boundary vertices.  We denote it by $\Gamma_U$.
Its oriented edge set is written as $E_U$, and its first Betti number as
$\beta(U)$.  The restriction of $A$ to $U$ defines edge integrals
\[
        a_I(A):=\int_I A\,\dd x,\qquad I\in E_U .
\]
After choosing a cycle basis of $\Gamma_U$, let $C_U$ be the corresponding
cycle matrix and let
\[
        \Phi_U(A)\in \R^{\beta(U)}/(2\pi\Z)^{\beta(U)}
\]
be the vector of fluxes through those cycles.  If $\beta(U)=0$, all formulas
below are read with the value zero.

\begin{definition}[Flux distance]\label{def:fluxdistance}
Let $U$ be a regular domain.  The flux distance of $A$ on $U$ is
\begin{equation}\label{eq:deltaU}
        \delta_A(U):=
        \min_{m\in\Z^{\beta(U)}}
        \inf\left\{\sum_{I\in E_U}|\eta_I|:
        \eta\in\R^{E_U},\ C_U\eta=-\widehat\Phi_U(A)+2\pi m\right\} ,
\end{equation}
where $\widehat\Phi_U(A)$ is any real representative of the flux vector.  The
value is independent of the representative.  It is also independent, up to the
obvious coordinate change, of the chosen cycle basis.
\end{definition}

The sign in \eqref{eq:deltaU} is a matter of convention and comes from the
choice $D_A=-\ii\partial_x-A$.  Replacing it by $+\widehat\Phi_U(A)$ gives the
same number after changing $m$.

\begin{definition}[Frustration index]\label{def:frustration}
For a regular domain $U$, define
\begin{equation}\label{eq:iotaA}
        \iota_A(U):=
        \inf_{\tau\in W^{1,1}(U;S^1)}\int_U |D_A\tau|\,\dd x .
\end{equation}
The phase $\tau$ is required to be continuous at the interior vertices of
$\Gamma_U$; no condition is imposed at boundary endpoints.
\end{definition}

\begin{proposition}[Frustration equals flux distance]\label{prop:frustration-flux}
For every regular domain $U\subset\Gamma$,
\begin{equation}\label{eq:iota-delta}
        \iota_A(U)=\delta_A(U).
\end{equation}
In particular, $\iota_A(U)=0$ if and only if all fluxes of $A$ through cycles
contained in $U$ belong to $2\pi\Z$.
\end{proposition}

\begin{proof}
It is enough to argue on each connected component of $U$, and we suppress this
minor notational point.  Let $\tau\in W^{1,1}(U;S^1)$.  On each edge interval
$I$ one may write $\tau=\ee^{\ii\theta_I}$ with
$\theta_I\in W^{1,1}(I;\R)$.  Then
\[
        D_A\tau=(\theta_I'-A_I)\tau,
        \qquad |D_A\tau|=|\theta_I'-A_I| .
\]
Set
\[
        \eta_I:=\int_I(\theta_I'-A_I)\,\dd x .
\]
For every cycle of $\Gamma_U$, the sum of the phase increments
$\int_I\theta_I'$ is an integral multiple of $2\pi$, because $\tau$ is
single-valued.  Hence $C_U\eta=-\widehat\Phi_U(A)+2\pi m$ for some
$m\in\Z^{\beta(U)}$.  Moreover,
\[
        \sum_{I\in E_U}|\eta_I|
        \leq \sum_{I\in E_U}\int_I|\theta_I'-A_I|\,\dd x
        =\int_U|D_A\tau|\,\dd x .
\]
Taking the infimum over $\tau$ gives $\delta_A(U)\leq\iota_A(U)$.

Conversely, fix $\eta\in\R^{E_U}$ satisfying
$C_U\eta=-\widehat\Phi_U(A)+2\pi m$.  We construct a phase whose magnetic
derivative has integral $\eta_I$ on each edge $I$.  Put
$a_I=\int_I A_I\,\dd x$ and prescribe the phase increment
$s_I:=\eta_I+a_I$ along $I$.  The cycle condition says precisely that the sum
of these increments along every cycle lies in $2\pi\Z$.  Therefore vertex
phases can be assigned consistently modulo $2\pi$: choose one vertex phase on a
spanning tree and extend along the tree, while the cycle condition gives
compatibility on the remaining edges.  On each interval $I=(0,\ell_I)$ choose
$\theta_I\in W^{1,\infty}(I)$ satisfying
\[
        \theta_I'(x)-A_I(x)=\frac{\eta_I}{\ell_I}
        \quad\text{a.e. on }I
\]
and having the prescribed endpoint phases.  Then $\tau_I=\ee^{\ii\theta_I}$ is
single-valued on $\Gamma_U$ and
\[
        \int_I |D_A\tau_I|\,\dd x
        =\int_I \left|\frac{\eta_I}{\ell_I}\right|\,\dd x
        =|\eta_I| .
\]
Thus $\iota_A(U)\leq\sum_I|\eta_I|$.  Taking the infimum over all admissible
$\eta$ proves the reverse inequality.
\end{proof}

\section{Magnetic Cheeger inequality}

\begin{definition}[Magnetic Cheeger constant]\label{def:magcheeger}
The magnetic Cheeger constant associated with $A$ is
\begin{equation}\label{eq:hA}
        h_A(\Gamma):=
        \inf_{U\in\cR(\Gamma)}
        \frac{|\partial U|+\delta_A(U)}{|U|},
\end{equation}
where $\cR(\Gamma)$ denotes the class of regular domains of positive length.
\end{definition}

The denominator in \eqref{eq:hA} is $|U|$, not
$\min\{|U|,|\Gamma\setminus U|\}$.  This is the one-sided Cheeger constant
which is appropriate for the bottom eigenvalue.  If the magnetic flux is
trivial, $U=\Gamma$ gives $h_A(\Gamma)=0$, in agreement with
$\lambda_0(H_A)=0$.

\paragraph{Phase regularity on level sets.}
We shall use the following elementary consequence of the one-dimensional
Sobolev chain rule.  If $u\in H^1(\Gamma)$ and $t>0$ is such that
$U_t=\{x\in\Gamma: |u(x)|>\sqrt t\}$ is a regular domain, then the phase
$\tau=u/|u|$ belongs to $W^{1,1}(U_t;S^1)$ and is continuous at the interior
vertices of $\Gamma_{U_t}$.  Indeed, on $U_t$ the map
$z\mapsto z/|z|$ is evaluated away from the origin.  The Sobolev chain rule
on each edge gives $\tau\in W^{1,1}$, with
$|\tau'|\leq C_t |u'|$ a.e. on $U_t$; since $U_t$ has finite length and
$u'\in L^2$, the right-hand side belongs to $L^1(U_t)$.  Vertex continuity
follows from the continuity of $u$ and from $|u|>\sqrt t$ at each interior
vertex of $U_t$.

\begin{theorem}[Magnetic Cheeger bound]\label{thm:cheeger}
Let $\Gamma$ be a finite compact connected metric graph and let
$A\in L^\infty(\Gamma)$ be real-valued.  Then
\begin{equation}\label{eq:cheeger-main}
        \lambda_0(H_A)\geq \frac{1}{8}\,h_A(\Gamma)^2 .
\end{equation}
\end{theorem}

\begin{proof}
Let $u\in H^1(\Gamma)$ be non-zero.  Write $\rho=|u|$.  For $t>0$ set
\[
        U_t:=\{x\in\Gamma:\rho(x)>\sqrt t\}.
\]
For almost every $t$, the set $U_t$ is a regular domain.  On $U_t$ the phase
$\tau=u/|u|$ belongs to $W^{1,1}(U_t;S^1)$ and hence, by
Proposition~\ref{prop:frustration-flux},
\begin{equation}\label{eq:frust-level}
        \delta_A(U_t)=\iota_A(U_t)
        \leq \int_{U_t}|D_A\tau|\,\dd x .
\end{equation}
On the set where $\rho>0$ one has the pointwise identity
\begin{equation}\label{eq:polar-energy}
        |D_Au|^2=|\rho'|^2+\rho^2|D_A\tau|^2 .
\end{equation}
The one-dimensional coarea formula on each edge gives
\begin{equation}\label{eq:coarea-boundary}
        \int_0^\infty |\partial U_t|\,\dd t
        \leq \int_\Gamma |(\rho^2)'|\,\dd x
        =2\int_\Gamma \rho |\rho'|\,\dd x .
\end{equation}
Furthermore, by Fubini and \eqref{eq:frust-level},
\begin{equation}\label{eq:coarea-frust}
        \int_0^\infty \delta_A(U_t)\,\dd t
        \leq \int_\Gamma \rho^2 |D_A\tau|\,\dd x .
\end{equation}
Combining \eqref{eq:polar-energy}--\eqref{eq:coarea-frust}, and using the
pointwise inequality
\[
        2ab+a^2c\leq 2\sqrt2\,a\sqrt{b^2+a^2c^2}
        \qquad (a,b,c\geq0),
\]
with $a=\rho$, $b=|\rho'|$ and $c=|D_A\tau|$, yields
\begin{equation}\label{eq:mag-coarea}
\begin{aligned}
        \int_0^\infty\bigl(|\partial U_t|+\delta_A(U_t)\bigr)\,\dd t
        &\leq 2\sqrt2\int_\Gamma \rho |D_Au|\,\dd x  \\
        &\leq 2\sqrt2\,\norm{u}_{L^2(\Gamma)}\,\qA[u]^{1/2} .
\end{aligned}
\end{equation}
On the other hand, the layer-cake representation gives
\[
        \int_0^\infty |U_t|\,\dd t
        =\int_\Gamma |u|^2\,\dd x
        =\norm{u}_{L^2(\Gamma)}^2 .
\]
By the definition of $h_A(\Gamma)$,
\[
        h_A(\Gamma)\norm{u}_{L^2}^2
        \leq \int_0^\infty\bigl(|\partial U_t|+\delta_A(U_t)\bigr)\,\dd t .
\]
Together with \eqref{eq:mag-coarea} this implies
\[
        \qA[u]\geq \frac{1}{8}h_A(\Gamma)^2
        \norm{u}_{L^2(\Gamma)}^2 .
\]
Taking the infimum over all non-zero $u\in H^1(\Gamma)$ proves
\eqref{eq:cheeger-main}.
\end{proof}

\begin{remark}[On the constant]
The factor $1/8$ is the same constant obtained by the elementary magnetic
coarea argument for closed manifolds in \cite{LangeLiuPeyerimhoffPost2015}.
It is not expected to be sharp on simple graphs.  For example, on a circle the
exact ground-state energy is the square of the distance of the flux from
$2\pi\Z$, divided by the square of the length.
\end{remark}

\section{Explicit flux lower bounds}

The constant $h_A(\Gamma)$ is already flux-explicit through
Proposition~\ref{prop:frustration-flux}.  A simpler lower bound follows by separating the
proper-domain part from the full graph.  Define
\begin{equation}\label{eq:hproper}
        h_{\rm pr}(\Gamma):=
        \inf_{\substack{U\in\cR(\Gamma)\\ U\neq \Gamma}}
        \frac{|\partial U|}{|U|} .
\end{equation}
If the graph consists of a single point this definition is void, but that case
is excluded by the assumption $\ell_e>0$.  Since every proper regular domain
has at least one boundary endpoint and $|U|\leq L_\Gamma$, one always has
\begin{equation}\label{eq:hproper-basic}
        h_{\rm pr}(\Gamma)\geq \frac{1}{L_\Gamma} .
\end{equation}
Let
\begin{equation}\label{eq:globaldelta}
        \delta_\Gamma(\Phi):=\delta_A(\Gamma).
\end{equation}
By Proposition~\ref{prop:frustration-flux}, this number depends only on the global flux
vector $\Phi(A)$.

\begin{proposition}[Global flux bound]\label{prop:global-flux-bound}
For every real $A\in L^\infty(\Gamma)$,
\begin{equation}\label{eq:h-global-bound}
        h_A(\Gamma)
        \geq
        \min\left\{h_{\rm pr}(\Gamma),\frac{\delta_\Gamma(\Phi)}{L_\Gamma}\right\} .
\end{equation}
Consequently,
\begin{equation}\label{eq:lambda-global-bound}
        \lambda_0(H_A)
        \geq \frac{1}{8}
        \min\left\{h_{\rm pr}(\Gamma),\frac{\delta_\Gamma(\Phi)}{L_\Gamma}\right\}^2 .
\end{equation}
\end{proposition}

\begin{proof}
If $U\neq\Gamma$, then
\[
        \frac{|\partial U|+\delta_A(U)}{|U|}
        \geq \frac{|\partial U|}{|U|}
        \geq h_{\rm pr}(\Gamma).
\]
If $U=\Gamma$, then $|\partial U|=0$, $|U|=L_\Gamma$ and
\[
        \frac{|\partial U|+\delta_A(U)}{|U|}
        =\frac{\delta_\Gamma(\Phi)}{L_\Gamma}.
\]
Taking the infimum in \eqref{eq:hA} gives \eqref{eq:h-global-bound}.  The
spectral estimate \eqref{eq:lambda-global-bound} follows from
\Cref{thm:cheeger}.
\end{proof}

The flux norm can be bounded below by a more elementary distance on the torus.
For a fixed cycle basis define
\begin{equation}\label{eq:d-infty}
        d_\infty(\Phi,0):=
        \min_{m\in\Z^\beta}\max_{1\leq j\leq\beta}
        |\widehat\Phi_j+2\pi m_j|,
\end{equation}
where $\widehat\Phi$ is any real representative.

\begin{corollary}[A coarse flux-only bound]\label{cor:coarse}
For every cycle basis,
\begin{equation}\label{eq:delta-dinfty}
        \delta_\Gamma(\Phi)\geq d_\infty(\Phi,0).
\end{equation}
Therefore
\begin{equation}\label{eq:lambda-coarse}
        \lambda_0(H_A)
        \geq \frac{1}{8L_\Gamma^2}
        \min\{1,d_\infty(\Phi,0)\}^2 .
\end{equation}
\end{corollary}

\begin{proof}
Let $C$ be the cycle matrix.  If $C\eta=-\widehat\Phi+2\pi m$, then for every
row $j$,
\[
        |\widehat\Phi_j-2\pi m_j|
        =\left|\sum_e C_{je}\eta_e\right|
        \leq \sum_e |\eta_e|.
\]
Taking the maximum over $j$, then the infimum over $\eta$ and $m$, gives
\eqref{eq:delta-dinfty}.  Combining this with
\eqref{eq:hproper-basic} and Proposition~\ref{prop:global-flux-bound} proves
\eqref{eq:lambda-coarse}.
\end{proof}

\begin{example}[Circle]\label{ex:circle}
Let $\Gamma$ be a circle of length $L$.  There is one flux
$\Phi\in\R/2\pi\Z$, and
\[
        \delta_\Gamma(\Phi)=\dist(\widehat\Phi,2\pi\Z).
\]
After a gauge transformation the operator is the magnetic Laplacian with
constant potential $\widehat\Phi/L$ on the periodic interval.  Hence
\begin{equation}\label{eq:circle-exact}
        \lambda_0(H_A)= \frac{\dist(\widehat\Phi,2\pi\Z)^2}{L^2} .
\end{equation}
The gap closes for integral flux and is maximal at half-periodic flux.  The
Cheeger bound gives the same periodic dependence, though with a non-sharp
constant.
\end{example}

\begin{example}[Theta graph]\label{ex:theta}
Let $\Gamma$ consist of three oriented edges $e_1,e_2,e_3$ joining the same two
vertices, with lengths $\ell_1,\ell_2,\ell_3$ and
$L_\Gamma=\ell_1+\ell_2+\ell_3$.  A convenient cycle basis is
$\gamma_1=e_1-e_2$ and $\gamma_2=e_1-e_3$.  Put
\[
        \phi_1=\int_{e_1}A-\int_{e_2}A,
        \qquad
        \phi_2=\int_{e_1}A-\int_{e_3}A
        \quad \text{modulo }2\pi .
\]
For real representatives $y_1=\phi_1+2\pi m_1$ and
$y_2=\phi_2+2\pi m_2$, the minimization in \eqref{eq:deltaU} becomes
\[
        \min_{t\in\R}\bigl(|t|+|t-y_1|+|t-y_2|\bigr).
\]
The minimum is attained at a median of $\{0,y_1,y_2\}$ and equals the length of
the smallest interval containing these three points.  Consequently
\begin{equation}\label{eq:theta-delta}
        \delta_\Gamma(\Phi)
        =\min_{m_1,m_2\in\Z}
        \left(
        \max\{0,y_1,y_2\}-\min\{0,y_1,y_2\}
        \right).
\end{equation}
For the same graph the proper-domain constant is
\begin{equation}\label{eq:theta-hpr}
        h_{\rm pr}(\Gamma)=\frac{2}{L_\Gamma}.
\end{equation}
Indeed, every proper regular domain has at least two boundary endpoints, while
cutting one interior point of an edge gives a proper domain of length
$L_\Gamma$ with two boundary endpoints, in the sense of the cut-graph
convention used above.  Combining \eqref{eq:theta-delta} with
\eqref{eq:lambda-global-bound} gives the fully explicit estimate
\begin{equation}\label{eq:theta-lambda}
        \lambda_0(H_A)
        \geq \frac{1}{8L_\Gamma^2}
        \min\left\{2,\delta_\Gamma(\Phi)\right\}^2 .
\end{equation}
This lower bound is expressed only in terms of the total length and the two
cycle fluxes.
\end{example}

\section{Magnetic heat decay}

Let
\[
        T_A(t):=\ee^{-tH_A},\qquad t\geq0,
\]
be the magnetic heat semigroup.  It is a contraction on $L^2(\Gamma)$, but it
is not positivity preserving in general.  The following estimates are simply
spectral consequences of the preceding lower bounds.

\begin{corollary}[Heat semigroup decay]\label{cor:heat}
For all $f\in L^2(\Gamma)$ and all $t\geq0$,
\begin{equation}\label{eq:heat-decay-h}
        \norm{T_A(t)f}_{L^2}
        \leq
        \exp\left(-\frac{t}{8}h_A(\Gamma)^2\right)
        \norm{f}_{L^2} .
\end{equation}
In particular,
\begin{equation}\label{eq:heat-decay-flux}
        \norm{T_A(t)f}_{L^2}
        \leq
        \exp\left[-\frac{t}{8}
        \min\left\{h_{\rm pr}(\Gamma),\frac{\delta_\Gamma(\Phi)}{L_\Gamma}\right\}^2
        \right]
        \norm{f}_{L^2} .
\end{equation}
\end{corollary}

\begin{proof}
By the spectral theorem,
$\norm{T_A(t)}_{2\to2}=\ee^{-t\lambda_0(H_A)}$.  The estimates follow from
\Cref{thm:cheeger} and Proposition~\ref{prop:global-flux-bound}.
\end{proof}

\begin{corollary}[Magnetic energy decay]\label{cor:energy-decay}
For every $f\in L^2(\Gamma)$ and $t>0$,
\begin{equation}\label{eq:energy-general-f}
        \norm{D_A T_A(t)f}_{L^2}^2
        \leq \frac{1}{e t}
        \exp\left(-\frac{t}{8}h_A(\Gamma)^2\right)
        \norm{f}_{L^2}^2 .
\end{equation}
If $f\in H^1(\Gamma)$, then for all $t\geq0$,
\begin{equation}\label{eq:energy-form-f}
        \norm{D_A T_A(t)f}_{L^2}^2
        \leq
        \exp\left(-\frac{t}{4}h_A(\Gamma)^2\right)
        \norm{D_A f}_{L^2}^2 .
\end{equation}
The same statements hold with $h_A(\Gamma)$ replaced by the lower bound in
\eqref{eq:h-global-bound}.
\end{corollary}

\begin{proof}
The form identity gives
\[
        \norm{D_A T_A(t)f}_{L^2}^2
        =\norm{H_A^{1/2}\ee^{-tH_A}f}_{L^2}^2 .
\]
For $f\in L^2$, the spectral theorem and the elementary inequality
$\lambda\ee^{-2t\lambda}\leq (et)^{-1}\ee^{-t\lambda_0(H_A)}$ for
$\lambda\geq\lambda_0(H_A)$ imply
\[
        \norm{H_A^{1/2}\ee^{-tH_A}f}_{L^2}^2
        \leq \frac{1}{et}\ee^{-t\lambda_0(H_A)}\norm{f}_{L^2}^2 .
\]
Using \Cref{thm:cheeger} proves \eqref{eq:energy-general-f}.  If
$f\in H^1(\Gamma)=\dom H_A^{1/2}$, then
\[
        \norm{H_A^{1/2}\ee^{-tH_A}f}_{L^2}^2
        \leq \ee^{-2t\lambda_0(H_A)}
        \norm{H_A^{1/2}f}_{L^2}^2,
\]
which yields \eqref{eq:energy-form-f} after another application of
\Cref{thm:cheeger}.
\end{proof}

\section*{Declarations}

\noindent\textbf{Competing interests.}
The author has no competing interests to declare.

\medskip
\noindent\textbf{Funding.}
No funding was received for this work.

\medskip
\noindent\textbf{Data availability.}
Data sharing is not applicable to this article as no datasets were generated or analysed.

\medskip
\noindent\textbf{Author contributions.}
The author conceived the argument, wrote the manuscript, and is responsible for the mathematical content.

\end{document}